\documentclass[12pt,draftcls,onecolumn]{IEEEtran}                                                          

\IEEEoverridecommandlockouts                              
\usepackage{graphics} 
\usepackage{epsfig} 
\usepackage{mathptmx} 
\usepackage{times} 
\usepackage{amsmath} 
\usepackage{amssymb}  
\usepackage{floatrow}

\newtheorem{theorem}{\indent Theorem}
\newtheorem{lemma}[theorem]{\indent Lemma}
\newtheorem{corollary}[theorem]{\indent Corollary}

\newtheorem{remark}{\indent Remark}

\title{\LARGE \bf
Sliding Mode Control of Two-Level Quantum Systems }

\author{Daoyi~Dong and Ian R. Petersen
\thanks{The material in this paper was partially presented at the 2010 American Control Conference
at Baltimore, Maryland, USA. This work was supported by the Australian Research Council and in part by the National Natural Science Foundation of China under Grant 60703083 and Grant 60805029.}
\thanks{D. Dong is with the School of
Engineering and Information Technology, University of New South
Wales at the Australian Defence Force Academy, Canberra, ACT 2600,
Australia and with the Institute of Cyber Systems and Control, State Key Laboratory
of Industrial Control Technology, Zhejiang University, Hangzhou 310027,
China
        {\tt\small daoyidong@gmail.com}}%
\thanks{I. R.
Petersen is with the School of Engineering and Information
Technology, University of New South Wales at the Australian Defence
Force Academy, Canberra, ACT 2600, Australia
        {\tt\small i.r.petersen@gmail.com}}%
}

\begin{document}

\maketitle
\thispagestyle{empty}
\pagestyle{empty}

\begin{abstract}

This paper proposes a robust control method based on sliding mode
design for two-level quantum systems with bounded uncertainties. An
eigenstate of the two-level quantum system is identified as a
sliding mode. The objective is to design a control law to steer the
system's state into the sliding mode domain and then maintain it in
that domain when bounded uncertainties exist in the system
Hamiltonian. We propose a controller design method using the
Lyapunov methodology and periodic projective measurements. In
particular, we give conditions for designing such a control law,
which can guarantee the desired robustness in the presence of the
uncertainties. The sliding mode control method has potential
applications to quantum information processing with uncertainties.

\end{abstract}


\section{INTRODUCTION}\label{Sec1}

The manipulation and control of quantum systems is becoming an
important task in many fields \cite{Dong and Petersen
2010IET}-\cite{Rabitz 2009}, such as atomic physics \cite{Chu 2002},
molecular chemistry \cite{Rabitz et al 2000}
and quantum information \cite{Nielsen and Chuang 2000}. It is
desirable to develop quantum control theory in a systematic way in
order to adapt it to the development of quantum technology
\cite{Dowling and Milburn 2003}. Several useful tools from classical
control theory have been introduced to the control analysis and
design of quantum systems. For example, optimal control theory has
been used to assist in control design for closed and dissipative
quantum systems \cite{Khaneja et al 2001}-\cite{Boscain and Mason
2006}. A learning control method has been presented for guiding the
control of chemical reactions \cite{Rabitz et al 2000}, \cite{Judson
and Rabitz 1992}. Quantum feedback control approaches including
measurement-based feedback and coherent feedback have been used to
improve performance for several classes of tasks such as preparing
quantum states, quantum error correction, controlling quantum
entanglement \cite{Wiseman and Milburn 1993}-\cite{Nurdin et al 2009
SIAM}. Robust control tools have been introduced to enhance the
robustness of quantum feedback networks and linear quantum
stochastic systems \cite{D'Helon and James 2006}, \cite{James et al
2007}.

Although some progress has been made, more research effort is
necessary in controlling quantum phenomena. In particular, the
robustness of quantum control systems has been recognized as a key
issue in developing practical quantum technology \cite{Pravia et al
2003}-\cite{Viola and Knill 2003}. In this paper, we focus on the
robustness problem for quantum control systems. In \cite{James et al
2007}, James and co-workers have formulated and solved a quantum
robust control problem using the $H^{\infty}$ method for linear
quantum stochastic systems. Here, we develop a variable structure
control approach with sliding modes to enhance the robustness of
quantum systems. The variable structure control strategy is a widely
used design method in classical control theory and industrial
applications where one can change the controller structure according
to a specified switching logic in order to obtain desired
closed-loop properties \cite{Utkin 1977}, \cite{Decarlo et al 1988}.
In \cite{Dong and Petersen 2008} and \cite{Dong and Petersen 2010
IJC}, Dong and Petersen have formulated and solved a variable
structure control problem for the control of quantum systems.
However, the results in \cite{Dong and Petersen 2008} only involve
open-loop control design using an idea of changing controller
structures and do not consider the robustness which can be obtained
through sliding mode control. Ref. \cite{Dong and Petersen 2008} and
Ref. \cite{Vilela Mendes and Man'ko 2003} have briefly discussed the
possible application of sliding mode control to quantum systems. In
\cite{Dong and Petersen 2009NJP}, two approaches based on sliding
mode design have been proposed for the control of quantum systems
and potential applications of sliding mode control to quantum
information processing have been presented. Following these results,
this paper formally presents a sliding mode control method for
two-level quantum systems to deal with bounded uncertainties in the
system Hamiltonian \cite{Dong and Petersen 2010ACC}. In particular,
we propose two approaches of designing the measurement period for
different situations which are dependent on the bound on the
uncertainties and the allowed probability of failure.

Variable structure control design with sliding modes generally
includes two main steps: selecting a sliding surface (sliding mode)
and controlling the system to and maintaining it in this sliding
surface. Being in the sliding surface guarantees that the quantum
system has the desired dynamics. We will select an eigenstate of the
free Hamiltonian of the controlled quantum system as a sliding mode.
In the second step, direct feedback control is not directly
applicable since we generally cannot acquire state information
without destroying the quantum system's state. Hence, we propose a
new method to accomplish this task, which is based on the Lyapunov
methodology and periodic projective measurements. The Lyapunov
methodology is a powerful tool for designing control laws in
classical control theory and has also been applied to quantum
control problems \cite{Mirrahimi et al 2005}-\cite{Kuang and Cong
2008}. Most existing results on Lyapunov control of quantum systems
focus on designing a control law to ensure that the controlled
quantum system's state asymptotically converges to the target state.
The existing Lyapunov design methods in quantum control rely on
perfect knowledge of the initial quantum states and system
Hamiltonian. In our approach, once the Lyapunov control steers the
quantum system into a sliding mode domain, we make a projective
measurement on the system. Hence, the Lyapunov design method can
tolerate small drifts (uncertainties) when carrying out our control
tasks, which will be demonstrated by simulation in Section II.C.
Periodic projective measurements are employed to maintain the
system's state in the sliding mode domain when uncertainties exist
in the system Hamiltonian. If the measurement period is small enough
and the initial state is an eigenstate, the frequent measurements
make the system collapse back to the initial state. This is related
to the quantum Zeno effect (for details, see \cite{Misra and
Sudarshan}, \cite{Itano et al 1990} and \cite{Facchi and Pascazio
2002}). In contrast to the quantum Zeno effect, our objective is to
design a measurement period which is as large as possible. The
framework of the proposed method involves unitary control (Lyapunov
control) and projective measurement. In this sense, it is similar to
the discrete-time quantum feedback stabilization problem in
\cite{Ticozzi and Bolognani 2010 MTNS} and \cite{Ticozzi and
Bolognani 2010}. However, these papers do not consider possible
uncertainties in the system Hamiltonian and use generalized
measurements rather than periodic projective measurements. The main
feature of the proposed method is that the control law can guarantee
control performance when bounded uncertainties exist in the system
Hamiltonian.

This paper is organized as follows. Section \ref{Sec2} introduces a
quantum control model, defines the sliding mode and formulates the
control problem. In Section \ref{Sec3}, we present a sliding mode
control method based on the Lyapunov methodology and periodic
projective measurements for two-level quantum systems with bounded
uncertainties. Using the known information about uncertainties
(e.g., the uncertainty bound and type of uncertainties), we propose
two approaches (i.e., Eqs. (\ref{1period}) and (\ref{2period})) for
designing the measurement period to guarantee the control
performance. An illustrative example is presented to demonstrate the
proposed method. The detailed proofs of the main theorems are
presented in Section \ref{Sec4}. Conclusions are given in Section
\ref{Sec5}.

\section{Sliding modes and problem formulation}\label{Sec2}
In this section, we first introduce a two-level quantum control
model. Then a sliding mode is defined using an eigenstate. Finally
the control problem considered in this paper is formulated.

\subsection{Quantum Control Model}
In this paper, we focus on two-level pure-state quantum systems. The
quantum state can be represented by a two-dimensional unit vector
$|\psi\rangle$ in a Hilbert space $\mathcal{H}$. Since the global
phase of a quantum state has no observable physical effect, we do
not consider the effect of global phase. If we denote the Pauli
matrices $\sigma=(\sigma_{x},\sigma_{y},\sigma_{z})$ as follows:
\begin{equation}
\sigma_{x}=\begin{pmatrix}
  0 & 1  \\
  1 & 0  \\
\end{pmatrix} , \ \ \ \
\sigma_{y}=\begin{pmatrix}
  0 & -i  \\
  i & 0  \\
\end{pmatrix} , \ \ \ \
\sigma_{z}=\begin{pmatrix}
  1 & 0  \\
  0 & -1  \\
\end{pmatrix} ,
\end{equation}
we may select the free Hamiltonian of the two-level quantum system
as $H_{0}=I_{z}=\frac{1}{2}\sigma_{z}$. Its two eigenstates are
denoted as $|0\rangle$ and $|1\rangle$. To control a quantum system,
we introduce the following control Hamiltonian
$H_{u}=\sum_{k}u_{k}(t)H_{k}$, where $u_{k}(t)\in \mathbf{R}$ and
$\{H_{k}\}$ is a set of time-independent Hamiltonians. For
simplicity, the control Hamiltonian for two-level systems can be
written as $H_{u}=u_{x}(t)I_{x}+u_{y}(t)I_{y}+u_{z}(t)I_{z}$, where
\begin{equation}
I_{x}=\frac{1}{2}\sigma_{x}=\frac{1}{2}\begin{pmatrix}
  0 & 1  \\
  1 & 0  \\
\end{pmatrix} , \ \ \
I_{y}=\frac{1}{2}\sigma_{y}=\frac{1}{2}\begin{pmatrix}
  0 & -i  \\
  i & 0  \\
\end{pmatrix} .
\end{equation}

The controlled dynamical equation can be described as (we have
assumed $\hbar=1$ by using atomic units in this paper)
\begin{equation}\label{controlled model}
\begin{array}{l} i|\dot{\psi}(t)\rangle
=H_{0}|\psi(t)\rangle+\sum_{k=x,y,z}u_{k}(t)I_{k}|\psi(t)\rangle, \\
|\psi(t=0)\rangle=|\psi_{0}\rangle .
\end{array}
\end{equation}
This control problem is converted into the following problem: given
an initial state and a target state, find a set of controls
$\{u_{k}(t)\}$ in (\ref{controlled model})
to drive the controlled system from the initial state to the target
state.

In practical applications, we often use the density operator (or
density matrix) $\rho$ to describe the quantum state of a quantum
system. For a pure state $|\psi\rangle$, the corresponding density
operator is $\rho\equiv |\psi\rangle \langle \psi|$. For a two-level
quantum system, the state $\rho$ can be represented in terms of the
Bloch vector
$\mathbf{r}=(x,y,z)=(\text{tr}\{\rho\sigma_{x}\},\text{tr}\{\rho\sigma_{y}\},\text{tr}\{\rho\sigma_{z}\})$:
\begin{equation}\label{eq4}
\rho=\frac{1}{2}(I+\mathbf{r}\cdot \sigma) .
\end{equation}
The evolution equation of $\rho$ can be written as
\begin{equation}
\dot\rho=-i[H, \rho]
\end{equation}
where $[A, B]=AB-BA$ and $H$ is the total system Hamiltonian.

After we represent the state $\rho$ with the Bloch vector, the pure
states of a two-level quantum system correspond to the surface of
the Bloch sphere, where $(x, y, z)=(\sin\theta \cos\varphi,
\sin\theta \sin\varphi, \cos\theta)$, $\theta \in [0,\pi]$, $\varphi
\in [0, 2\pi]$. An arbitrary pure state $|\psi\rangle$ for a
two-level quantum system can be represented as
\begin{equation}\label{superposition}
|\psi\rangle=\cos{\frac{\theta}{2}}|0\rangle+e^{i\varphi}\sin{\frac{\theta}{2}}|1\rangle
.
\end{equation}

\subsection{Sliding Modes}
Sliding modes play an important role in variable structure control
\cite{Utkin 1977}. Usually, the sliding mode is constructed so that
the system has desired dynamics in the sliding surface. For a
quantum control problem, a sliding mode may be represented as a
functional of the state $|\psi\rangle$ and the Hamiltonian $H$;
i.e., $S(|\psi\rangle, H)=0$.
For example, an eigenstate $|\phi_{j}\rangle$ of the free
Hamiltonian $H_{0}$ (i.e.,
$H_{0}|\phi_{j}\rangle=\lambda_{j}|\phi_{j}\rangle$ where
$\lambda_{j}$ is one eigenvalue of $H_{0}$) can be selected as a
sliding mode. We can define $S(|\psi\rangle, H)=1-|\langle
\psi|\phi_{j}\rangle|^{2}=0$. If the initial state
$|\psi_{0}\rangle$ is in the sliding mode; i.e.,
$S(|\psi_{0}\rangle, H)=1-|\langle \psi_{0}|\phi_{j}\rangle|^{2}=0$,
we can easily prove that the quantum system will maintain its state
in this surface under only the action of the free Hamiltonian
$H_{0}$. In fact, $|\psi(t)\rangle=e^{-iH_{0}t}|\psi_{0}\rangle$,
and we have
\begin{equation}
\begin{array}{ll}
S(|\psi(t),H)&=1-|\langle \psi(t)|\phi_{j}\rangle|^{2}=1-|\langle \psi_{0}|e^{iH_{0}t}|\phi_{j}\rangle|^{2}\\
&=1-|\langle \psi_{0}|\phi_{j}\rangle
e^{i\lambda_{j}t}|^{2}=1-|\langle
\psi_{0}|\phi_{j}\rangle|^{2} |e^{i\lambda_{j}t}|^{2}\\
&=0 .\\
\end{array}
\end{equation}
That is, an eigenstate of $H_{0}$ can be identified as a sliding
mode. For two-level quantum systems, we may select either
$|0\rangle$ or $|1\rangle$ as a sliding mode. Without loss of
generality, we identify the eigenstate $|0\rangle$ of a two-level
quantum system as the sliding mode in this paper.

\subsection{Problem Formulation}
In Section II.B, we have identified an eigenstate $|0\rangle$ as a
sliding mode. This means that if a quantum system is driven into the
sliding mode, its state will be maintained in the sliding surface
under the action of the free Hamiltonian. However, in practical
applications, it is inevitable that there exist noises and
uncertainties. In this paper, we suppose that the uncertainties can
be approximately described as perturbations in the Hamiltonian. That
is, the uncertainties can be denoted as
$H_{\Delta}=\epsilon_{x}(t)I_{x}+\epsilon_{y}(t)I_{y}+\epsilon_{z}(t)I_{z}$.
The unitary errors in \cite{Pravia et al 2003} belong to this class
of uncertainties and uncertainties in one-qubit (one quantum bit)
gate also correspond to this class of uncertainties \cite{Dong and
Petersen 2009NJP}. For a spin system in solid-state nuclear magnetic
resonance (NMR), external noisy magnetic fields and unwanted
coupling with other spins may lead to uncertainties in this class.
Further, we assume the uncertainties are bounded; i.e.,
\begin{equation}
\sqrt{\epsilon_{x}^{2}(t)+\epsilon_{y}^{2}(t)+\epsilon_{z}^{2}(t)}\leq
\epsilon\ \ (\epsilon \geq 0).
\end{equation}
When $\epsilon=0$, $H_{\Delta}=0$. That is, there exist no
uncertainties, which is trivial for our problem. Hence, in the
following we assume $\epsilon
>0$. An important advantage of classical sliding mode control is its
robustness. Our main motivation for introducing sliding mode control
to quantum systems is to deal with uncertainties. We further suppose
that the corresponding system without uncertainties is completely
controllable and arbitrary unitary control operations can be
generated. This assumption can be guaranteed for a two-level quantum
system if we can realize arbitrary rotations along the $z$-axis and
$\zeta$-axis ($\zeta=x\ \text{or}\ y$) (e.g., see \cite{D'Alessandro
2007} for details).

The control problem under consideration is stated as follows: design
a control law to drive and then maintain the quantum system's state
in a sliding mode domain even when bounded uncertainties exist in
the system Hamiltonian. Here a sliding mode domain may be defined as
$\mathcal{D}=\{|\psi\rangle :  |\langle 0|\psi\rangle|^{2}\geq
1-p_{0}, 0< p_{0}< 1\}$, where $p_{0}$ is a given constant. Here we
assume $p_{0}\neq 0, 1$ since the case $p_{0}=0$ only occurs in the
sliding mode surface and the case $p_{0}=1$ is always true. Hence,
the two cases with $p_{0}=0$ and $p_{0}=1$ are trivial for our
problem. The definition of the sliding mode domain implies that the
system has a probability of at most $p_{0}$ (which we call the
probability of failure) to collapse out of $\mathcal{D}$ when making
a measurement. This behavior is quite different from that which
occurs in traditional sliding mode control. Hence, we expect that
our control laws will guarantee that the system's state remains in
$\mathcal{D}$ except that a measurement operation may take it away
from $\mathcal{D}$ with a small probability (not greater than
$p_{0}$). The control problem considered in this paper includes
three main subtasks: (i) for any initial state (assumed to be
known), design a control law to drive the system's state into a
defined sliding mode domain $\mathcal{D}$; (ii) design a control law
to maintain the system's state in $\mathcal{D}$; (iii) design a
control law to drive the system's state back to $\mathcal{D}$ if a
measurement operation takes it away from $\mathcal{D}$. For
simplicity, we suppose that there exist no uncertainties during the
control processes (i) and (iii).

\section{Sliding mode control based on Lyapunov methods and projective measurements}\label{Sec3}

\subsection{General Method}
The first task is to design a control law to drive the controlled
system to the chosen sliding mode domain $\mathcal{D}$.
Lyapunov-based methods are widely used to accomplish this task in
traditional sliding mode control. If the gradient of a Lyapunov
function is negative in the neighborhood of the sliding surface,
then the controlled system's state will be attracted to and
maintained in $\mathcal{D}$. The Lyapunov methodology has also been
used to design control laws for quantum systems \cite{Mirrahimi et
al 2005}-\cite{Kuang and Cong 2008}. However, these existing results
do not consider the issue of robustness against uncertainties. Since
the measurement of a quantum system will inevitably destroy the
measured state, most existing results on Lyapunov-based control for
quantum systems in fact use a feedback design to construct an
open-loop control. That is, Lyapunov-based control can be used to
first design a feedback law which is then used to find the open-loop
control by simulating the closed-loop system. Then the control can
be applied to the quantum system in an open-loop way. Hence, the
traditional sliding mode control methods using Lyapunov control
cannot be directly applied to our problem.

Although quantum measurement often has deleterious effects in
quantum control tasks, recent results have shown that it can be
combined with unitary transformations to complete some quantum
manipulation tasks and enhance the capability of quantum control
\cite{Vilela Mendes and Man'ko 2003}, \cite{Pechen et al
2006}-\cite{Romano and D'Alessandro 2006}. For example, Vilela
Mendes and Man'ko \cite{Vilela Mendes and Man'ko 2003} showed
nonunitarily controllable systems can be made controllable by using
``measurement plus evolution". Quantum measurement can be used as a
control tool as well as a method of information acquisition. It is
worth mentioning that the effect of measurement on a quantum system
as a control tool can be achieved through the interaction between
the system and measurement apparatus. In this paper, we will combine
the Lyapunov methodology and projective measurements (with the
measurement operator $\sigma_{z}$) to accomplish the sliding mode
control task for two-level quantum systems. The projective measurement with $\sigma_{z}$ on a two-level
system makes the system's state collapse into $|0\rangle$ (corresponding to eigenvalue $1$ of $\sigma_{z}$) or $|1\rangle$ (corresponding to eigenvalue $-1$ of $\sigma_{z}$).

The steps in the control algorithm are as follows (see Fig. 1):
\begin{enumerate}
    \item
    Select an eigenstate $|0\rangle$ of $H_{0}$ as a sliding mode $S(|\psi\rangle,
    H)=0$, and define the sliding mode domain as $\mathcal{D}=\{|\psi\rangle
    :  |\langle 0|\psi\rangle|^{2}\geq 1-p_{0}\}$.
    \item For a known initial state $|\psi_{0}\rangle$, construct a Lyapunov function $V(|\psi_{0}\rangle,
S)$ to find the control law that can drive $|\psi_{0}\rangle$ into
the sliding mode $S$.
    \item For a specified probability of failure $p_{0}$ and $V(|\psi_{0}\rangle,
S)$, construct the control period $T_{0}$ so that the control law
can drive the system's state into $\mathcal{D}$ in a time period
$T_{0}$.
    \item For an initial condition which is another eigenstate
$|1\rangle$, design a Lyapunov function $V(|1\rangle, S)$ and
construct the period $T_{1}$ by using a similar method to that in
3).
    \item According to $p_{0}$ and $\epsilon$, design the period $T$ for periodic projective
    measurements.
    \item Use the designed control law to drive the system's state into $\mathcal{D}$ in
    $T_{0}$ and make a projective measurement at $t=T_{0}$.
    Then repeat the following operations: make periodic projective
    measurements with the period $T$ to maintain the system's state in
    $\mathcal{D}$;
    if the measurement result corresponds to $|1\rangle$, we use the
    corresponding control law to drive the state back into $\mathcal{D}$.
\end{enumerate}

\begin{figure}
\centering
\includegraphics[width=4.6in]{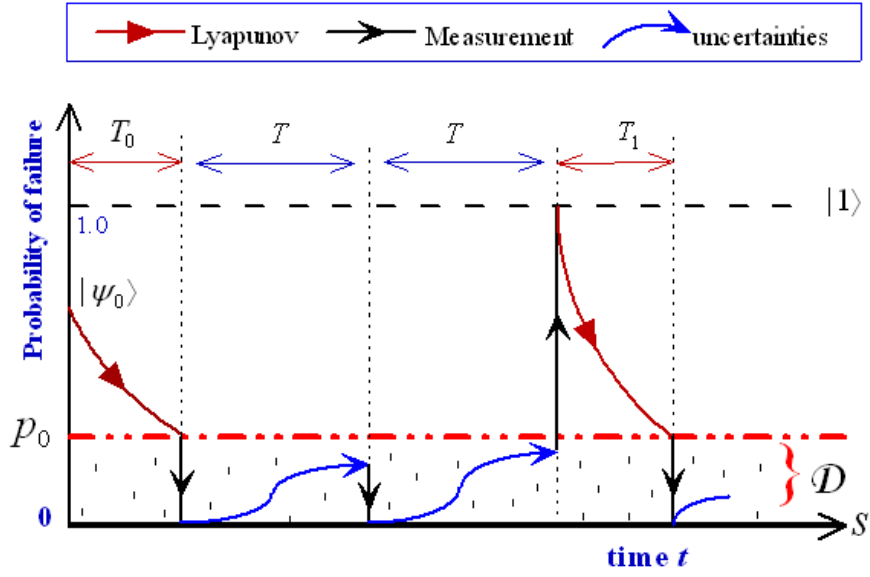}
\caption{The sliding mode control scheme for a two-level quantum
system based on Lyapunov methods and periodic projective
measurements. In this figure, ``Lyapunov", ``Measurement" and
``uncertainties" mean the evolution process of the quantum system
under the Lyapunov control law, the projective measurement and
uncertainties in the system Hamiltonian, respectively.}
\label{scheme1}
\end{figure}

From the above control algorithm, we see that the design of Lyapunov
functions and the selection of the period $T$ for the projective
measurements are the two most important tasks. To design a control
law for quantum systems, several Lyapunov functions have been
constructed, such as state distance-based and average value-based
approaches \cite{Mirrahimi et al 2005}-\cite{Kuang and Cong 2008}.
Here we select a function based on the Hilbert-Schmidt distance
between a state $|\psi\rangle$ and the sliding mode state
$|\phi_{j}\rangle$ as a Lyapunov function \cite{Grivopoulos and
Bamieh 2003}, \cite{Kuang and Cong 2008}; i.e.,
$$V(|\psi\rangle, S)=\frac{1}{2}(1-|\langle
\phi_{j}|\psi\rangle|^{2}).$$ It is clear that $V\geq 0$. The
first-order time derivative of $V$ is
\begin{equation}
\begin{array}{ll}
\dot{V}& =-\sum_{k=x,y,z}u_{k}\Im[\langle \psi|\phi_{j}\rangle
\langle \phi_{j}|I_{k}|\psi\rangle]\\
& =-\sum_{k=x,y,z}u_{k}|\langle \psi|\phi_{j}\rangle|\Im[e^{i\angle
\langle \psi|\phi_{j}\rangle} \langle \phi_{j}|I_{k}|\psi\rangle],\\
\end{array}
\end{equation}
where $\Im[a+bi]=b$ ($a, b\in \mathbf{R}$) and $\angle c$ denotes
the argument of a complex number $c$. To ensure $\dot{V}\leq 0$, we
choose the control laws as in \cite{Kuang and Cong 2008}:
\begin{equation}\label{Lyapunov_control}
 u_{k}=K_{k}f_{k}(\Im[e^{i\angle
\langle \psi|\phi_{j}\rangle} \langle \phi_{j}|I_{k}|\psi\rangle]),
\ \ \ (k=x,y,z)
\end{equation}
where $K_{k}>0$ may be used to adjust the control amplitude and
$f(\cdot)$ satisfies $xf(x)\geq 0$. Define $\angle \langle
\psi|\phi_{j}\rangle=0^{\circ}$ when $\langle
\psi|\phi_{j}\rangle=0$.

When one employs a Lyapunov methodology to design a control law,
LaSalle's invariance principle is a useful tool to analyze its
convergence. That is, if $\dot{x}(t)=g(x(t))$ is an autonomous
dynamical system with phase space $\Omega$ and $V(x)$ is a Lyapunov
function on $\Omega$ satisfying $V(x)> 0$ for all $x\neq x_{0}$ and
$\dot{V}(x(t))\leq 0$, any bounded solution converges to the
invariant set $E=\{x|\dot{V}(x(t))=0\}$ as $t\rightarrow +\infty$
(for details, see \cite{LaSalle and Lefschetz 1961}). For two-level
quantum systems, LaSalle's invariance principle can guarantee that
the quantum state converges to the sliding mode $|0\rangle$ under
the control law in (\ref{Lyapunov_control}) (for details, see
\cite{Kuang and Cong 2008}). The convergence is asymptotic. Hence,
we make a projective measurement with the measurement operator
$\sigma_{z}$ when we apply the Lyapunov control to the system for
$T_{0}$ (corresponding to the initial condition $|\psi_{0}\rangle$)
or $T_{1}$ (corresponding to the initial condition $|1\rangle$),
which will drive the system into $|0\rangle$ with a probability not
less than $1-p_{0}$.

Another important task is to design the measurement period $T$. We
can estimate a bound according to the bound $\epsilon$ on the
uncertainties and the allowed probability of failure $p_{0}$. Then,
we construct a period $T$ to guarantee control performance according
to the estimated bound. An extreme case is $T\rightarrow 0$. That
is, after the quantum system's state is driven into the sliding
mode, we make frequent measurements. This corresponds to the quantum
Zeno effect \cite{Itano et al 1990}, which is the inhibition of
transitions between quantum states by frequent measurement of the
state (see, e.g., \cite{Misra and Sudarshan} and \cite{Itano et al
1990}). Frequent measurements (i.e., $T\rightarrow 0$) can guarantee
that the state is maintained in the sliding mode in spite of the
presence of uncertainties. However, it is usually a difficult task
to make such frequent measurements. We may conclude that the smaller
$T$ is, the bigger the cost of accomplishing the periodic
measurements becomes. Hence, in contrast to the quantum Zeno effect,
we wish to design a measurement period $T$ as large as possible. In
the following subsection, we will propose two approaches of
designing $T$ for different situations.


\subsection{The Design of the Measurement Period $T$}
We select the sliding mode as $S(|\psi\rangle, H)=1-|\langle
\psi|0\rangle|^2=0$. If there exist no uncertainties and we have
driven the system's state to the sliding mode at time $t_{0}$, it
will be maintained in the sliding mode using only the free
Hamiltonian $H_{0}$; i.e., $S(|\psi_{(t\geq t_{0})}\rangle,
H_{0})=0$. That is, if the quantum system's state is driven into the
sliding mode, it will evolve in the sliding surface. However, in
practical applications, some uncertainties are unavoidable, which
may drive the system's state away from the sliding mode. We wish to
design a control law to ensure the desired robustness in the
presence of uncertainties. Assume that the state at time $t$ is
$\rho_{t}$. If we make measurements on this system, the probability
$p$ that it will collapse into $|1\rangle$ (the probability of
failure) is
\begin{equation}\label{fidelity}
p=\langle 1|\rho_{t}|1\rangle=\frac{1-z_{t}}{2} ,
\end{equation}
where $z_{t}=\text{tr}(\rho_{t}\sigma_{z})$. We have assumed that
the possible uncertainties can be described by
$H_{\Delta}=\epsilon_{x}(t)I_{x}+\epsilon_{y}(t)I_{y}+\epsilon_{z}(t)I_{z}$,
where unknown $\epsilon_{x}(t)$, $\epsilon_{y}(t)$ and
$\epsilon_{z}(t)$ satisfy
$\sqrt{\epsilon_{x}^{2}(t)+\epsilon_{y}^{2}(t)+\epsilon_{z}^{2}(t)}\leq
\epsilon$. We now give detailed discussions to design the
measurement period $T$ for possible uncertainties.

First we consider a special case $H_{\Delta}=\epsilon(t)I_{z}$
($|\epsilon(t)|\leq \epsilon$). This case corresponds to phase-flip
type bounded uncertainties. For any $H_{\Delta}=\epsilon_{z}I_{z}$
(where $|\epsilon_{z}|\leq \epsilon$), if $S(|\psi_{0}\rangle,
H)=0$, we have
\begin{equation}
\begin{array}{ll}
S(|\psi(t)\rangle,H)& =1-|\langle \psi(t)|0\rangle|^{2}\\
& =1-|\langle
\psi_{0}|e^{i(H_{0}+\epsilon_{z}I_{z})t}|0\rangle|^{2}\\
& =1-|\langle \psi_{0}|0\rangle|^{2}
|e^{\frac{i(1+\epsilon_{z})}{2}t}|^{2}\\
& =0 .
\end{array}
\end{equation}
This type of uncertainty does not drive the system's state away from
the sliding mode. Hence we ignore this type of uncertainty in our
method.

Now we consider the unknown uncertainties
$H_{\Delta}=\epsilon_{x}(t)I_{x}+\epsilon_{y}(t)I_{y}$ (where
$\sqrt{\epsilon_{x}^{2}(t)+\epsilon_{y}^{2}(t)} \leq \epsilon$) and
have the following theorem.
\begin{theorem}
For a two-level quantum system with the initial state
$|\psi(0)\rangle=|0\rangle$ at the time $t=0$, the system evolves to
$|\psi(t)\rangle$ under the action of $H(t)=I_{z}+\epsilon_{x}(t)
I_{x}+\epsilon_{y}(t) I_{y}$ (where
$\sqrt{\epsilon_{x}^{2}(t)+\epsilon_{y}^{2}(t)}\leq \epsilon$ and
$\epsilon >0$). If $t\in [0, T^{(1)}]$, where
\begin{equation}\label{1period}
T^{(1)}=\frac{\arccos(1-2p_{0})}{\epsilon},
\end{equation}
the system's state will remain in $\mathcal{D}=\{|\psi\rangle :
|\langle 0 |\psi\rangle|^{2}\geq 1-p_{0}\}$ (where $0< p_{0}< 1$).
When one makes a projective measurement with the measurement
operator $\sigma_{z}$ at the time $t$, the probability of failure
$p=|\langle 1|\psi(t)\rangle|^{2}$ is not greater than $p_{0}$.
\end{theorem}

Using Theorem 1, we may try to maintain the system's state in
$\mathcal{D}$ (i.e., the subtask (ii)) by implementing periodic
projective measurements with the measurement period $T=T^{(1)}$. If
we have more knowledge about the uncertainties, it is possible to
improve the measurement period $T^{(1)}$. Now assume that the
uncertainty is $H_{\Delta}=\epsilon(t)I_{\zeta}$ ($\zeta=x\
\text{or}\ y$) and $p_{0}\in (0,
\frac{\epsilon^{2}}{1+\epsilon^{2}}]$. We have the following
theorem.

\begin{theorem}
For a two-level quantum system with the initial state
$|\psi(0)\rangle=|0\rangle$ at the time $t=0$, the system evolves to
$|\psi(t)\rangle$ under the action of
$H(t)=I_{z}+\epsilon(t)I_{\zeta}$ (where $\zeta=x\ \text{or}\ y$,
$|\epsilon(t)|\leq \epsilon$ and $\epsilon
>0$). If $p_{0}\in (0, \frac{\epsilon^{2}}{1+\epsilon^{2}}]$ and
$t\in [0, T^{(2)}]$, where
\begin{equation}\label{2period}
T^{(2)}=\frac{\arccos[1-2(1+\frac{1}{\epsilon^{2}})p_{0}]}{\sqrt{1+\epsilon^{2}}},
\end{equation}
the system's state will remain in $\mathcal{D}=\{|\psi\rangle :
|\langle 0 |\psi\rangle|^{2}\geq 1-p_{0}\}$ (where $0< p_{0}< 1$).
When one makes a projective measurement with the measurement
operator $\sigma_{z}$ at the time $t$, the probability of failure
$p=|\langle 1|\psi(t)\rangle|^{2}$ is not greater than $p_{0}$.
\end{theorem}

\begin{remark}
The proofs of Theorem 1 and Theorem 2 will be presented in Section
\ref{Sec4}. The two theorems mean the following fact. For a
two-level quantum system with unknown uncertainties
$H_{\Delta}=\epsilon_{x}(t)I_{x}+\epsilon_{y}(t)I_{y}$ (where
$\sqrt{\epsilon_{x}^{2}(t)+\epsilon_{y}^{2}(t)}\leq \epsilon$), if
its initial state is in the sliding mode $|0\rangle$, we can ensure
that the probability of failure is not greater than a given constant
$p_{0}$ ($0< p_{0}< 1$) through implementing periodic projective
measurements with the measurement period $T=T^{(1)}$ using
(\ref{1period}). Further, if we know that $p_{0}$ and $\epsilon$
satisfy the relationship $0<p_{0}\leq
\frac{\epsilon^{2}}{1+\epsilon^{2}}$ and there exists only one type
of uncertainty (i.e., $H_{\Delta}=\epsilon(t)I_{x}$ or
$H_{\Delta}=\epsilon(t)I_{y}$, where $|\epsilon(t)|\leq \epsilon$),
we can design a measurement period $T=T^{(2)}$ using (\ref{2period})
which is larger than $T^{(1)}$. The proof of Theorem 2 also shows
that $T^{(2)}$ is an optimal measurement period. This measurement
period will guarantee the required robustness. It is easy to prove
the relationship $T^{(2)}\geq T^{(1)}$ for arbitrary $p_{0}\in (0,
\frac{\epsilon^{2}}{1+\epsilon^{2}}]$. The detailed proof will be
presented in the Appendix. Based on the above analysis, the
selection rule for $T$ is summarized in Table I. Moreover, from
(\ref{1period}) and (\ref{2period}), it is clear that for a constant
$\epsilon$, $T^{(1)}\rightarrow 0$ and $T^{(2)}\rightarrow 0$ when
$p_{0}\rightarrow 0$. That is, for a given bound $\epsilon$ on the
uncertainties, if we expect to guarantee the probability of failure
$p_{0}\rightarrow 0$, it requires us to implement frequent
measurements such that the measurement period $T\rightarrow 0$.
Another special case is $\epsilon\rightarrow +\infty$, which leads
to $T^{(1)}\rightarrow 0$ and $T^{(2)}\rightarrow 0$. That is, to
deal with very large uncertainties, we need to make frequent
measurements ($T\rightarrow 0$) to guarantee the desired robustness.
From (\ref{1period}), we also know that for a given $p_{0}$,
$T^{(1)}$ monotonically decreases with increasing $\epsilon$. This
means that we need to employ a smaller measurement period to deal
with uncertainties with a larger bound $\epsilon$.
\end{remark}
\begin{table}[!hbp]
\centering \caption{A summary on the selection rule of the
measurement period $T$}
\begin{tabular}{|c|c|c|}
\hline
Type of uncertainties & $H_{\Delta}=\epsilon_{x}(t)I_{x}+\epsilon_{y}(t)I_{y}$
& $H_{\Delta}=\epsilon(t)I_{\zeta}$ ($\zeta=x\ \text{or}\ y)$ \\
\hline Allowed probability of failure $p_{0}$ & $0< p_{0}< 1$ &
\begin{tabular}{c|c} $0< p_{0}\leq
\frac{\epsilon^{2}}{1+\epsilon^{2}}$  &
 $\frac{\epsilon^{2}}{1+\epsilon^{2}}< p_{0}< 1$ \end{tabular}  \\
\hline The measurement period $T$ & $T=T^{(1)}$ &
\begin{tabular}{c|c}
$T=T^{(2)}$ \ \ \
& \ \ \ $T=T^{(1)}$  \end{tabular}   \\
\hline
\end{tabular}
\end{table}
\subsection{An Illustrative Example}
Now we present an illustrative example to demonstrate the proposed
method. Assume $p_{0}=0.01$. Consider two cases: (a)
$\epsilon=0.02$; (b) $\epsilon=0.2$. For simplicity, we assume
$|\psi_{0}\rangle=|1\rangle$. Hence, $T_{0}=T_{1}$. We first design
the control and $T_{1}$ using (\ref{Lyapunov_control}). Here, we
consider control only using $H_{u}=\frac{1}{2}u(t)\sigma_{y}$. Using
(\ref{Lyapunov_control}), we select $u(t)=K(\Im[e^{i\angle \langle
\psi(t)|0\rangle} \langle 0|\sigma_{y}|\psi(t)\rangle])$ and
$K=100$. Let the time stepsize be given by $\delta t=10^{-4}$. We
can obtain the probability curve of $|0\rangle$ shown in Fig.
\ref{SMC_probability}, the control value shown in Fig.
\ref{SMC_control} and $T_{1}=0.060$. For $\epsilon=0.02$, we can
design the measurement period $T=T^{(1)}=10.017$ using
(\ref{1period}). For $\epsilon=0.2$, we can design the measurement
period $T=T^{(1)}=1.002$ using (\ref{1period}). Since
$p'=\frac{\epsilon^{2}}{1+\epsilon^{2}}=3.8 \times 10^{-2}>p_{0}$
when $\epsilon=0.2$, if the uncertainties take the form of
$H_{\Delta}=\epsilon(t)I_{\zeta}$ ($\zeta=x\ \text{or}\ y)$, we can
improve the measurement period to $T=T^{(2)}=1.049$ using
(\ref{2period}). It is clear that $T\gg T_{1}$ in these two cases.
For some practical quantum systems such as spin systems in NMR, we
can use strong control actions (e.g., $K=10^{5}$) to drive the
system from $|1\rangle$ into $\mathcal{D}$ within a short time
period $T_{1}$ \cite{Khaneja et al 2001}. These facts make the
assumption of no uncertainties in the control process reasonable.
Moreover, the fact that the measurement period $T$ is much greater
than the control time required to go to $|0\rangle$ from $|1\rangle$
indicates the possibility of realizing such a periodic measurement
on a practical quantum system.

\begin{figure}
\centering
\includegraphics[width=3.6in]{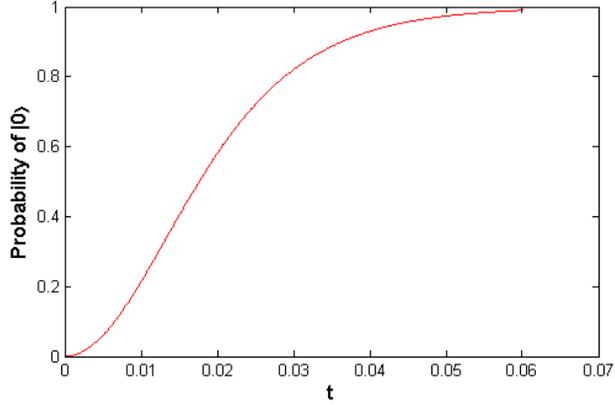}
\caption{The probability of $|0\rangle$ under Lyapunov control.}
\label{SMC_probability}
\end{figure}

\begin{figure}
\centering
\includegraphics[width=3.6in]{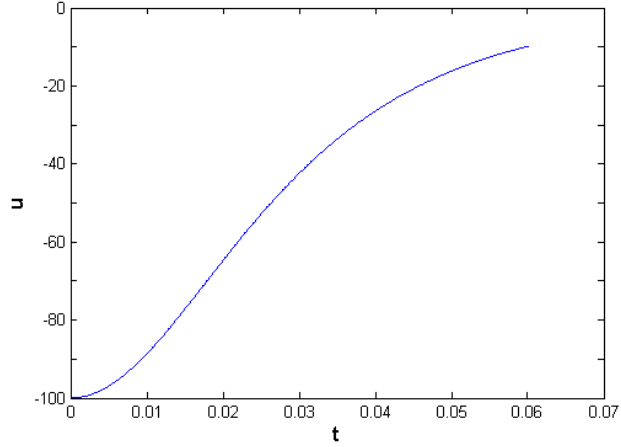}
\caption{The control value $u(t)$.} \label{SMC_control}
\end{figure}

\begin{remark}
In the process of designing the control law for driving the system's
state from $|1\rangle$ to $|0\rangle$, we employ an approach based
on the Lyapunov methodology. An advantage of such an approach is
that it is relatively easy to find a control law by simulation. It
is worth noting that most existing applications of the Lyapunov
methodology to quantum systems do not involve measurement. Here, we
combine the Lyapunov-based control and projective measurements for
controlling quantum systems, which in some applications make our
method more useful than the Lyapunov-based control for quantum
systems proposed in previous papers. In \cite{Dong and Petersen
2009NJP}, an approach based on time-optimal control design has also
been proposed to complete this task. The advantage of such an
approach is that we take the shortest time to complete the control
task. However, it is generally difficult to find a complete
time-optimal solution for high-dimensional quantum systems. For the
above simple task, it has been proven that the time-optimal control
employs a bang-bang control strategy \cite{Boscain and Mason 2006}.
Using the method in \cite{Boscain and Mason 2006}, we should take
$u=-100$ in $t \in [0,\ 0.016]$ and then use $u=100$ in $t\in
(0.016,\ 0.030]$. In this case, the total time required is
$T'_{1}=0.030$ ($< T_{1}=0.060$).
\end{remark}

\begin{remark}
In the process of designing the Lyapunov control for driving the
system's state from $|1\rangle$ to $|0\rangle$, we ignore possible
uncertainties. By simulation, we find that small uncertainties can
also be tolerated in this process. For example, if $\epsilon=0.02$
and the uncertainty $\epsilon(t)$ is the noise with a uniform
distribution on the interval $[-0.02, 0.02]$, the probability curves
of $|0\rangle$ are shown in Fig. \ref{figure_with_noise} when we
apply the control obtained from Fig. \ref{SMC_control} to the
quantum system. The probabilities of $|0\rangle$ for the cases with
uncertainties are very close to the probability of $|0\rangle$ for
the case without uncertainties. By more simulation, we find that the
final probability of $|0\rangle$ is $(99.00\pm 0.01)\%$ for
$\epsilon(t)I_x$ ($|\epsilon(t)|\leq \epsilon$ where $\epsilon=0.02$
or $0.2$), the final probability of $|0\rangle$ is $(99.00\pm
0.02)\%$ for $\epsilon(t)I_y$ ($|\epsilon(t)|\leq 0.02$) and the
final probability of $|0\rangle$ is $(99.00\pm 0.13)\%$ for
$\epsilon(t)I_y$ ($|\epsilon(t)|\leq 0.2$). If we use a smaller
probability of failure $\tilde{p}_{0}$ (e.g., $\tilde{p}_{0}=0.5
p_{0}$) as the terminal condition of the Lyapunov control or employ
a bigger $K$ for the same $T_{1}$, these simulations suggest that it
is possible to ensure that the Lyapunov control will drive the
system's state into the sliding mode domain even when there exist
small uncertainties.

\begin{figure}
\centering
\includegraphics[width=6.0in]{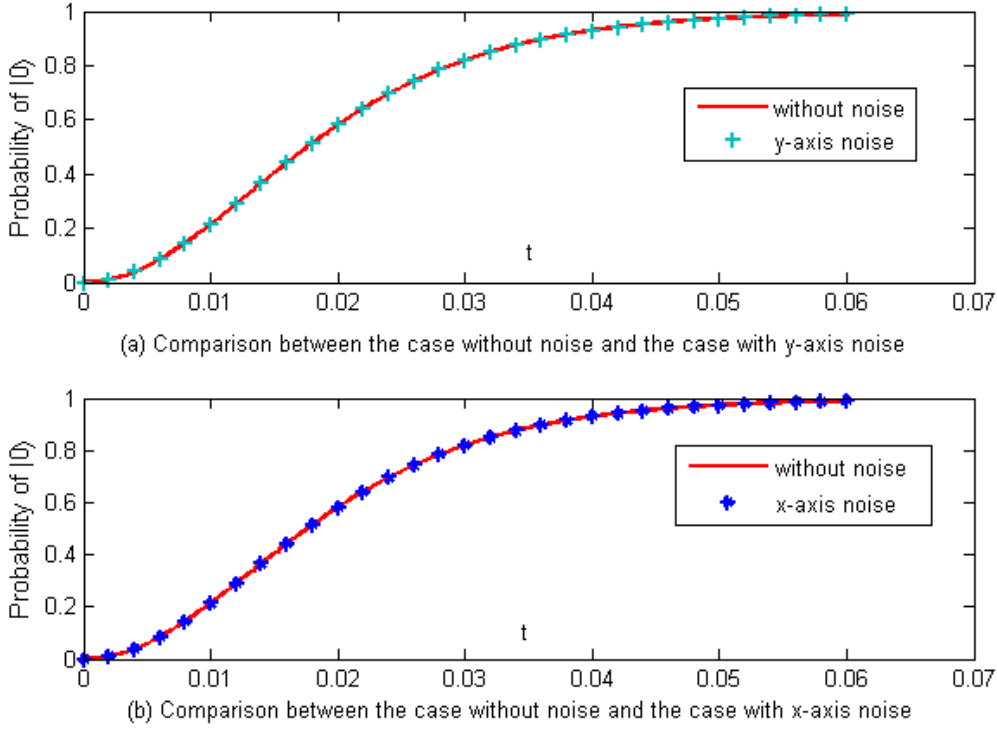}
\caption{The probability curves of $|0\rangle$ for the case without
uncertainties (without noise) and the cases with uncertainties
(noise) when we apply the control in Fig. \ref{SMC_control} to the
quantum system. The $\xi$-axis noise ($\xi=x$ or $y$) means the
existence of $\epsilon(t)I_{\xi}$ where $\epsilon(t)$ is the noise
with a uniform distribution on the interval $[-0.02,
0.02]$.}\label{figure_with_noise}
\end{figure}
\end{remark}

\section{Proof of Theorems}\label{Sec4}

This section will present the detailed proofs of Theorem 1 and
Theorem 2. The proof of Theorem 1 involves the following steps: (I)
Compare the probabilities of failure for
$H=I_{z}+\epsilon_{0}\cos\gamma_{0}I_{x}+\epsilon_{0}\sin\gamma_{0}I_{y}$
and
$H=\epsilon_{0}\cos\gamma_{0}I_{x}+\epsilon_{0}\sin\gamma_{0}I_{y}$
($\epsilon_{0}$ and $\gamma_{0}$ are constant); (II) Compare the
probabilities of failure for $H=\epsilon I_{x}$ and
$H=I_{z}+\epsilon(t) I_{x}$ ($|\epsilon(t)|\leq \epsilon$); (III)
Use the previous results to compare the probabilities of failure for
$H=\epsilon I_{x}$ and $H=I_{z}+\epsilon_{x}(t)I_{x}+\epsilon_{y}(t)
I_{y}$ ($\sqrt{\epsilon_{x}^{2}(t)+\epsilon_{y}^{2}(t)}\leq
\epsilon$); (IV) Use $H=\epsilon I_{x}$ to estimate the measurement
period. The basic steps for the proof of Theorem 2 include: (I)
Formulate the problem of finding the ``worst" case as an optimal
control problem of $\min z_{f}$; (II) Obtain the optimal control
solution for nonsingular cases; (III) Exclude the possibility of
singular cases; (IV) Use the ``worst" case to estimate the
measurement period. Considering that the arguments in the proof of
Theorem 2 are useful for the proof of Theorem 1, we will first
present the proof of Theorem 2 and then prove Theorem 1.

\subsection{Proof of Theorem 2}
\begin{proof}
For $H^{A}=I_{z}+\epsilon(t)I_{x}$, using $\dot{\rho}=-i[H^{A},
\rho]$ and (\ref{eq4}), we have
\begin{equation}
\left(%
\begin{array}{cc}
  \dot{z}_{t} & \dot{x}_{t}-i\dot{y}_{t} \\
  \dot{x}_{t}+i\dot{y}_{t} & -\dot{z}_{t} \\
\end{array}%
\right)
=\left(%
\begin{array}{cc}
  \epsilon(t) y_{t} & -y_{t}-ix_{t}+i\epsilon(t) z_{t} \\
  -y_{t}+ix_{t}-i\epsilon(t) z_{t} & -\epsilon(t) y_{t} \\
\end{array}%
\right).
\end{equation}
That is,
\begin{equation}\label{Thm2stateEq}
\left(%
\begin{array}{c}
  \dot{x}_{t} \\
  \dot{y}_{t} \\
  \dot{z}_{t} \\
\end{array}%
\right)
=\left(%
\begin{array}{ccc}
  0 & -1 & 0 \\
  1 & 0  & -\epsilon(t) \\
  0 & \epsilon(t) & 0 \\
\end{array}%
\right) \left(%
\begin{array}{c}
  x_{t} \\
  y_{t} \\
  z_{t} \\
\end{array}%
\right),
\end{equation}
where $(x_{0}, y_{0}, z_{0})=(0, 0, 1)$.

We now consider $\epsilon(t)$ as a control input and select the
performance measure as
\begin{equation}
J(\epsilon)=z_{f} .
\end{equation}
From (\ref{fidelity}), we know that the ``worst" case (i.e., the
case maximizing the probability of failure) corresponds to
minimizing $z_{f}$. Also, we introduce the Lagrange multiplier
vector $\lambda(t)=(\lambda_{1}(t), \lambda_{2}(t),
\lambda_{3}(t))^{T}$ and obtain the corresponding Hamiltonian
function as follows:
\begin{equation}
\mathbb{H}({\mathbf{r}(t),\epsilon(t),\mathbf{\lambda}(t),t})\equiv
\lambda^{T}(t)\left(%
\begin{array}{ccc}
  0 & -1 & 0 \\
  1 & 0  & -\epsilon(t) \\
  0 & \epsilon(t) & 0 \\
\end{array}%
\right) \left(%
\begin{array}{c}
  x_{t} \\
  y_{t} \\
  z_{t} \\
\end{array}%
\right),
\end{equation}
where $\mathbf{r}(t)=(x_{t}, y_{t}, z_{t})$. That is
\begin{equation}
\mathbb{H}({\mathbf{r}(t),\epsilon(t),\mathbf{\lambda}(t),t})
=-\lambda_{1}(t)y_{t}+\lambda_{2}(t)x_{t}+\epsilon(t)(\lambda_{3}(t)y_{t}-\lambda_{2}(t)z_{t}).
\end{equation}
According to Pontryagin's minimum principle \cite{Kirk 1970}, a
necessary condition for $\epsilon^{*}(t)$ to minimize
$J(\epsilon)$ is
\begin{equation}\label{necessary_condition}
\mathbb{H}({\mathbf{r}^{*}(t),\epsilon^{*}(t),\mathbf{\lambda}^{*}(t),t})\leq
\mathbb{H}({\mathbf{r}^{*}(t),\epsilon(t),\mathbf{\lambda}^{*}(t),t})
.
\end{equation}
The necessary condition provides a relationship to determine the
optimal control $\epsilon^{*}(t)$. If there exists a time interval
$[t_{1}, t_{2}]$ of finite duration during which the necessary
condition (\ref{necessary_condition}) provides no information about
the relationship between $\mathbf{r}^{*}(t)$, $\epsilon^{*}(t)$,
$\mathbf{\lambda}^{*}(t)$, we call the interval $[t_{1}, t_{2}]$ a
singular interval \cite{Kirk 1970}. If we do not consider singular
cases (i.e., $\lambda_{3}(t)y_{t}-\lambda_{2}(t)z_{t}\equiv 0$), the
optimal control $\epsilon^{*}(t)$ should be chosen as follows:
\begin{equation}\label{bangbang}
\epsilon^{*}(t)=-\epsilon
\text{sgn}(\lambda_{3}(t)y_{t}-\lambda_{2}(t)z_{t}).
\end{equation}
That is, the optimal control strategy for $\epsilon(t)$ is bang-bang
control; i.e., $\epsilon^{*}(t)=\bar{\epsilon}=+\epsilon \ \text{or}
-\epsilon$. Now we consider $H^{B}=I_{z}+\bar{\epsilon}I_{x}$ which
leads to the state equation

\begin{equation}\label{eq20}
\left(%
\begin{array}{c}
  \dot{x}_{t} \\
  \dot{y}_{t} \\
  \dot{z}_{t} \\
\end{array}%
\right)
=\left(%
\begin{array}{ccc}
  0 & -1 & 0 \\
  1 & 0  & -\bar{\epsilon} \\
  0 & \bar{\epsilon} & 0 \\
\end{array}%
\right) \left(%
\begin{array}{c}
  x_{t} \\
  y_{t} \\
  z_{t} \\
\end{array}%
\right),
\end{equation}
where $(x_{0}, y_{0}, z_{0})=(0, 0, 1)$.
The corresponding solution is
\begin{equation}\label{theorem1-bound-solution}
\left(%
\begin{array}{c}
  x_{t} \\
  y_{t} \\
  z_{t} \\
\end{array}%
\right)
=\left(%
\begin{array}{c}
  -\frac{\bar{\epsilon}}{1+\epsilon^{2}}\cos\omega t+\frac{\bar{\epsilon}}{1+\epsilon^{2}} \\
  -\frac{\bar{\epsilon}}{\sqrt{1+\epsilon^{2}}}\sin\omega t \\
  \frac{\epsilon^{2}}{1+\epsilon^{2}}\cos\omega t+\frac{1}{1+\epsilon^{2}} \\
\end{array}%
\right)
\end{equation}
where $\omega=\sqrt{1+\epsilon^{2}}$. From
(\ref{theorem1-bound-solution}), we know that $z_t$ is a
monotonically decreasing function in $t$ when $t\in [0,
\frac{\pi}{\sqrt{1+\epsilon^{2}}}]$. Hence, we only consider the
case $t \in [0,t_{f}]$ where $t_{f}\in [0,
\frac{\pi}{\sqrt{1+\epsilon^{2}}}]$.

Now consider the optimal control problem with a fixed final time
$t_{f}$ and a free final state $\mathbf{r}_{f}=(x_{f},y_{f},z_{f})$.
According to Pontryagin's minimum principle,
$\lambda^{*}(t_f)=\frac{\partial}{\partial
\mathbf{r}}\mathbf{r}^{*}(t_{f})$.
From this, it is straightforward to verify that
$(\lambda_{1}(t_f),\lambda_{2}(t_f),\lambda_{3}(t_f))=(0,0,1)$. Now
let us consider another necessary condition
$\dot{\lambda}(t)=-\frac{\partial
\mathbb{H}({\mathbf{r}(t),\epsilon(t),\mathbf{\lambda}(t),t})}{\partial
\mathbf{r}}$ which leads to the following relationships:

\begin{equation}\label{eq24b}
\dot{\mathbf{\lambda}}(t)=\left(%
\begin{array}{c}
  \dot{\lambda}_{1}(t) \\
  \dot{\lambda}_{2}(t) \\
  \dot{\lambda}_{3}(t) \\
\end{array}%
\right)
=\left(%
\begin{array}{ccc}
  0 & -1 & 0 \\
  1 & 0  & -\bar{\epsilon} \\
  0 & \bar{\epsilon} & 0 \\
\end{array}%
\right) \left(%
\begin{array}{c}
  \lambda_{1}(t) \\
  \lambda_{2}(t) \\
  \lambda_{3}(t) \\
\end{array}%
\right),
\end{equation}
where
$(\lambda_{1}(t_f),\lambda_{2}(t_f),\lambda_{3}(t_f))=(0,0,1)$. The
corresponding solution is
\begin{equation}
\left(%
\begin{array}{c}
  \lambda_{1}(t) \\
  \lambda_{2}(t) \\
  \lambda_{3}(t) \\
\end{array}%
\right)
=\left(%
\begin{array}{c}
  -\frac{\bar{\epsilon}}{1+\epsilon^{2}}\cos\omega(t_f-t)+\frac{\bar{\epsilon}}{1+\epsilon^{2}} \\
  \frac{\bar{\epsilon}}{\sqrt{1+\epsilon^{2}}}\sin\omega(t_f-t) \\
  \frac{\epsilon^{2}}{1+\epsilon^{2}}\cos\omega(t_f-t)+\frac{1}{1+\epsilon^{2}} \\
\end{array}%
\right).
\end{equation}

We obtain
\begin{equation}\label{singularcondition}
\lambda_{3}(t)y_{t}-\lambda_{2}(t)z_{t}
=\frac{-\bar{\epsilon}}{\omega^{3}}[\sin\omega
t+\epsilon^{2}\sin\omega t_f +\sin\omega(t_f-t)].
\end{equation}
It is easy to show that the quantity
$(\lambda_{3}(t)y_{t}-\lambda_{2}(t)z_{t})$ occurring in
(\ref{bangbang}) does not change its sign when $t_{f}\in [0,
\frac{\pi}{\sqrt{1+\epsilon^{2}}}]$ and $t \in [0,t_{f}]$.

Now we further exclude the possibility that there exists a singular
case. Suppose that there exists a singular interval $[t_{0}, t_{1}]$
(where $t_{0}\geq 0$) such that when $t\in [t_{0}, t_{1}]$
\begin{equation}\label{singular1Eq}
h(t)=\lambda_{3}(t)y_{t}-\lambda_{2}(t)z_{t}\equiv 0.
\end{equation}
We also have the following relationship
\begin{equation}\label{singular2Eq}
\dot{h}(t)=\lambda_{3}(t)x_{t}-\lambda_{1}(t)z_{t}\equiv 0
\end{equation}
where we have used (\ref{Thm2stateEq}) and the following costate
equation
\begin{equation}
\dot{\mathbf{\lambda}}(t)=\left(%
\begin{array}{c}
  \dot{\lambda}_{1}(t) \\
  \dot{\lambda}_{2}(t) \\
  \dot{\lambda}_{3}(t) \\
\end{array}%
\right)
=\left(%
\begin{array}{ccc}
  0 & -1 & 0 \\
  1 & 0  & -\epsilon(t) \\
  0 & \epsilon(t) & 0 \\
\end{array}%
\right) \left(%
\begin{array}{c}
  \lambda_{1}(t) \\
  \lambda_{2}(t) \\
  \lambda_{3}(t) \\
\end{array}%
\right).
\end{equation}

If $t_{0}=0$, we have $(x_{0}, y_{0}, z_{0})=(0, 0, 1)$. By the
principle of optimality \cite{Kirk 1970}, we may consider the case
$t_f=t_1$. Using (\ref{singular1Eq}), (\ref{singular2Eq}) and
$(\lambda_{1}(t_1),\lambda_{2}(t_1),\lambda_{3}(t_1))=(0,0,1)$, we
have $x_{t_1}=0$ and $y_{t_1}=0$. Using the relationship of
$x^{2}_{t}+y^{2}_{t}+z_{t}^{2}=1$, we obtain $z_{t_{1}}=1$ or $-1$.
If $z_{t_{1}}=1$, the initial state and the final state are the same
state $|0\rangle$. However, if we use the control
$\epsilon(t)=\bar{\epsilon}$, from (\ref{theorem1-bound-solution})
we have
$z_{t_1}(\bar{\epsilon})=\frac{\epsilon^2}{1+\epsilon^2}\cos\omega
t_1+\frac{1}{1+\epsilon^2}<z_{t_1}=1$. Hence, this contradicts the
fact that we are considering the optimal case $\min z_f$. If
$z_{t_{1}}=-1$, there exists $0< \tilde{t}_1< t_1$ such that
$z_{\tilde{t}_1}=0$. By the principle of optimality \cite{Kirk
1970}, we may consider the case $t_f=\tilde{t}_1$. From the two
equations (\ref{singular1Eq}) and (\ref{singular2Eq}), we know that
$z^{2}_{\tilde{t}_1}=1$ which contradicts $z_{\tilde{t}_1}=0$.
Hence, no singular condition can exist if $t_{0}=0$.

If $t_{0}> 0$, using (\ref{bangbang}) we must select
$\epsilon(t)=\bar{\epsilon}$ when $t\in [0, t_{0}]$. From
(\ref{singularcondition}), we know that there exist no $t_{0}\in (0,
t_{f})$ satisfying
$\lambda_{3}(t_{0})y_{t_{0}}-\lambda_{2}(t_{0})z_{t_{0}}=0$. Hence,
there exist no singular cases for our problem.

From the above analysis, $\epsilon(t)=\bar{\epsilon}$ is the optimal
control when $t\in [0, \frac{\pi}{\sqrt{1+\epsilon^{2}}}]$. Hence
$z^{A}_{t}=z_{t}(\epsilon(t))\geq z_{t}(\bar{\epsilon})=z^{B}_{t}$.
From (\ref{fidelity}), it is clear that the probabilities of failure
satisfy $p^{A}_{t}=\frac{1-z^{A}_{t}}{2}\leq
p^{B}_{t}=\frac{1-z^{B}_{t}}{2}$. That is, the probability of
failure $p_{t}^{A}$ is not greater than $p^{B}_{t}$ for $t\in [0,
\frac{\pi}{\sqrt{1+\epsilon^{2}}}]$. When $t\in [0,
\frac{\pi}{\sqrt{1+\epsilon^{2}}}]$, $z_{t}^{B}$ is monotonically
decreasing and $p_{t}^{B}$ is monotonically increasing. When
$t=\frac{\pi}{\sqrt{1+\epsilon^{2}}}$, using
(\ref{theorem1-bound-solution}) we have
$z^{B}_{t}=\frac{1-\epsilon^{2}}{1+\epsilon^{2}}$. That is, the
probability of failure $p'=\frac{\epsilon^{2}}{1+\epsilon^{2}}$.
Hence, we can design the measurement period $T$ using the case of
$H^{B}$ when $0< p_{0}\leq \frac{\epsilon^{2}}{1+\epsilon^{2}}$.

Using (\ref{fidelity}) and (\ref{theorem1-bound-solution}), for
$t\in [0, \frac{\pi}{\sqrt{1+\epsilon^{2}}}]$ we obtain the
probability of failure
\begin{equation}\label{period1}
p_{t}^{B}=\frac{\epsilon^{2}}{1+\epsilon^{2}}\frac{{1-\cos\omega
t}}{2} .
\end{equation}
Hence, we can design the maximum measurement period as follows
\begin{equation}
T^{(2)}=\frac{\arccos[1-2(1+\frac{1}{\epsilon^{2}})p_{0}]}{\sqrt{1+\epsilon^{2}}},
\end{equation}

For $H_{\Delta}=\epsilon(t)I_{y}$ (where $|\epsilon(t)|\leq
\epsilon$), we can obtain the same conclusion as that in the case
$H_{\Delta}=\epsilon(t)I_{x}$ (where $|\epsilon(t)|\leq \epsilon$).
\end{proof}

\subsection{Proof of Theorem 1}
To prove Theorem 1, we first prove two lemmas (Lemma \ref{lemmaB}
and Lemma \ref{lemma2}).

\begin{lemma}\label{lemmaB}
For a two-level quantum system with the initial state
$(x_{0},y_{0},z_{0})=(0,0,1)$ (i.e., $|0\rangle$), the system
evolves to $(x^{A}_{t},y^{A}_{t},z^{A}_{t})$ and
$(x^{B}_{t},y^{B}_{t},z^{B}_{t})$ under the action of
$H^{A}=I_{z}+\epsilon_{0}\cos\gamma_{0}I_{x}+\epsilon_{0}\sin\gamma_{0}I_{y}$
($\epsilon_{0}$ is a nonzero constant) and
$H^{B}=\epsilon_{0}\cos\gamma_{0}I_{x}+\epsilon_{0}\sin\gamma_{0}I_{y}$,
respectively. For arbitrary $t\in [0, \frac{\pi}{|\epsilon_{0}|}]$,
$z^{A}_{t}\geq z^{B}_{t}$.
\end{lemma}

\begin{proof}
For the system with Hamiltonian
$H^{A}=I_{z}+\epsilon_{0}\cos\gamma_{0}I_{x}+\epsilon_{0}\sin\gamma_{0}I_{y}$,
using $\dot{\rho}=-i[H, \rho]$ and (\ref{eq4}), we obtain the
following state equations
\begin{equation}\label{lemmaB1}
\left(%
\begin{array}{c}
  \dot{x}_{t}^{A} \\
  \dot{y}_{t}^{A} \\
  \dot{z}_{t}^{A} \\
\end{array}%
\right)
=\left(%
\begin{array}{ccc}
  0 & -1 & \epsilon_{0}\sin\gamma_{0} \\
  1 & 0  & -\epsilon_{0}\cos\gamma_{0} \\
  -\epsilon_{0}\sin\gamma_{0} & \epsilon_{0}\cos\gamma_{0} & 0 \\
\end{array}%
\right) \left(%
\begin{array}{c}
  x_{t}^{A} \\
  y_{t}^{A} \\
  z_{t}^{A} \\
\end{array}%
\right),
\end{equation}
where $(x_{0}^{A}, y_{0}^{A}, z_{0}^{A})=(0, 0, 1)$. The
corresponding solution is as follows
\begin{equation}
\left(%
\begin{array}{c}
  x_{t}^{A} \\
  y_{t}^{A} \\
  z_{t}^{A} \\
\end{array}%
\right)
=\\
\left(%
\begin{array}{c}
  \frac{\epsilon_{0}\sin\gamma_{0}}{\sqrt{1+\epsilon_{0}^{2}}}\sin\omega_{0} t-\frac{\epsilon_{0}\cos\gamma_{0}}{1+\epsilon_{0}^{2}}\cos\omega_{0} t+\frac{\epsilon_{0}\cos\gamma_{0}}{1+\epsilon_{0}^{2}} \\
  -\frac{\epsilon_{0}\cos\gamma_{0}}{\sqrt{1+\epsilon_{0}^{2}}}\sin\omega_{0}
  t-\frac{\epsilon_{0}\sin\gamma_{0}}{1+\epsilon_{0}^{2}}\cos\omega_{0}
  t+\frac{\epsilon_{0}\sin\gamma_{0}}{1+\epsilon_{0}^{2}} \\
  \frac{\epsilon_{0}^{2}}{1+\epsilon_{0}^{2}}\cos\omega_{0} t+\frac{1}{1+\epsilon_{0}^{2}} \\
\end{array}%
\right).
\end{equation}
where $\omega_{0}=\sqrt{1+\epsilon_{0}^{2}}$.

For the system with Hamiltonian
$H^{B}=\epsilon_{0}\cos\gamma_{0}I_{x}+\epsilon_{0}\sin\gamma_{0}I_{y}$,
using $\dot{\rho}=-i[H, \rho]$ and (\ref{eq4}), we obtain the
following state equations
\begin{equation}\label{lemmaB1}
\left(%
\begin{array}{c}
  \dot{x}_{t}^{B} \\
  \dot{y}_{t}^{B} \\
  \dot{z}_{t}^{B} \\
\end{array}%
\right)
=\left(%
\begin{array}{ccc}
  0 & 0 & \epsilon_{0}\sin\gamma_{0} \\
  0 & 0  & -\epsilon_{0}\cos\gamma_{0} \\
  -\epsilon_{0}\sin\gamma_{0} & \epsilon_{0}\cos\gamma_{0} & 0 \\
\end{array}%
\right) \left(%
\begin{array}{c}
  x_{t}^{B} \\
  y_{t}^{B} \\
  z_{t}^{B} \\
\end{array}%
\right),
\end{equation}
where $(x_{0}^{B}, y_{0}^{B}, z_{0}^{B})=(0, 0, 1)$. We can obtain
the corresponding solution as
\begin{equation}
\left(%
\begin{array}{c}
  x_{t}^{B} \\
  y_{t}^{B} \\
  z_{t}^{B} \\
\end{array}%
\right)
=\\
\left(%
\begin{array}{c}
  \sin\gamma_{0}\sin\epsilon_{0} t \\
  -\cos\gamma_{0}\sin\epsilon_{0}
  t \\
  \cos\epsilon_{0} t \\
\end{array}%
\right).
\end{equation}

Since $\cos\epsilon_0 t=\cos(-\epsilon_0 t)$, we may first assume
$\epsilon_0>0$. We define $F(t)$ and $f(t)$ as follows.
\begin{equation}
F(t)=z_{t}^{A}-z_{t}^{B},
\end{equation}

\begin{equation}\label{eq3}
f(t)=F'(t)=\epsilon_{0}\sin\epsilon_{0}t-\frac{\epsilon_{0}^{2}}{\sqrt{1+\epsilon_{0}^{2}}}\sin\omega_{0}
t .
\end{equation}

Now we consider $t\in [0, \frac{\pi}{\sqrt{1+\epsilon_{0}^{2}}}]$,
and obtain
\begin{equation}\label{eq5}
f'(t)=\epsilon_{0}^{2}(\cos\epsilon_{0}t-\cos\omega_{0}
t)=2\epsilon_{0}^{2}\sin\frac{\omega_{0}+\epsilon_{0}}{2}t\sin\frac{\omega_{0}-\epsilon_{0}}{2}t\geq
0
\end{equation}
It is clear that $f'(t)=0$ only when $t=0$. Hence $f(t)$ is a
monotonically increasing function and $$\min_t{f(t)}=f(0)=0 .$$

Hence, we have
\begin{equation}\label{eq6}
f(t)\geq 0 .
\end{equation}
From this, it is clear that $F(t)$ is a monotonically increasing
function and
$$\min_t{F(t)}=F(0)=0 .$$

Hence $F(t)\geq 0$ when $t\in [0,
\frac{\pi}{\sqrt{1+\epsilon_{0}^{2}}}]$. Moreover, it is clear that
$\min_t{z_{t}^{A}}=z_{t}^{A}|_{t=\frac{\pi}{\sqrt{1+\epsilon_{0}^{2}}}}$
and $z_{t}^{B}=\cos\epsilon_{0} t$ is a monotonically decreasing
function when $t\in [0, \frac{\pi}{|\epsilon_{0}|}]$. It is easy to
obtain the following relationship
$$z_{t}^{B}|_{t \in [\frac{\pi}{\sqrt{1+\epsilon_{0}^{2}}},
\frac{\pi}{|\epsilon_{0}|}]}\leq
z_{t}^{B}|_{t=\frac{\pi}{\sqrt{1+\epsilon_{0}^{2}}}}<z_{t}^{A}|_{t=\frac{\pi}{\sqrt{1+\epsilon_{0}^{2}}}}.$$

Hence we can conclude that $z^{A}_{t}\geq z^{B}_{t}$ for arbitrary
$t\in [0, \frac{\pi}{|\epsilon_{0}|}]$.

\end{proof}

Let $\gamma_{0}=0$, and we have the following corollary.
\begin{corollary}\label{corollary}
For a two-level quantum system with the initial state
$(x_{0},y_{0},z_{0})=(0,0,1)$ (i.e., $|0\rangle$), the system
evolves to $(x^{A}_{t},y^{A}_{t},z^{A}_{t})$ and
$(x^{B}_{t},y^{B}_{t},z^{B}_{t})$ under the action of
$H^{A}=I_{z}+\epsilon_{0}I_{x}$ (where $\epsilon_{0}$ is a nonzero
constant) and $H^{B}=\epsilon_{0}I_{x}$, respectively. For arbitrary
$t\in [0, \frac{\pi}{|\epsilon_{0}|}]$, $z^{A}_{t}\geq z^{B}_{t}$.
\end{corollary}

We now present another lemma.

\begin{lemma}\label{lemma2}
For a two-level quantum system with the initial state
$(x_{0},y_{0},z_{0})=(0,0,1)$ (i.e., $|0\rangle$), the system
evolves to $(x^{A}_{t},y^{A}_{t},z^{A}_{t})$ and
$(x^{B}_{t},y^{B}_{t},z^{B}_{t})$ under the action of
$H^{A}=I_{z}+\epsilon(t)I_{x}$ (where $|\epsilon(t)|\leq \epsilon$)
and $H^{B}=\epsilon I_{x}$, respectively. For arbitrary $t\in [0,
\frac{\pi}{\epsilon}]$, $z^{A}_{t}\geq z^{B}_{t}$.
\end{lemma}

\begin{proof}
First, we take an arbitrary evolution state (except $|1\rangle$)
starting from $|0\rangle$ as a new initial state. For
$H^{B}=\epsilon I_x$, the initial state can be represented as
$(x'_{0}, y'_{0}, z_{0})=(0, -\sin\theta_{0}, \cos\theta_{0})$,
where $\theta_{0} \in [0, \pi)$. We have
\begin{equation}\label{eq20}
\left(%
\begin{array}{c}
  \dot{x}_{t}^{B} \\
  \dot{y}_{t}^{B} \\
  \dot{z}_{t}^{B} \\
\end{array}%
\right)
=\left(%
\begin{array}{ccc}
  0 & 0 & 0 \\
  0 & 0  & -\epsilon \\
  0 & \epsilon & 0 \\
\end{array}%
\right) \left(%
\begin{array}{c}
  x_{t}^{B} \\
  y_{t}^{B} \\
  z_{t}^{B} \\
\end{array}%
\right).
\end{equation}
The corresponding solution is as follows:
\begin{equation}\label{bound}
\left(%
\begin{array}{c}
  x_{t}^{B} \\
  y_{t}^{B} \\
  z_{t}^{B} \\
\end{array}%
\right)
=\left(%
\begin{array}{c}
  0 \\
  -z_{0}\sin\epsilon t+y'_{0}\cos\epsilon t \\
  z_{0}\cos\epsilon t+y'_{0}\sin\epsilon t \\
\end{array}%
\right).
\end{equation}

For $H^{A}=I_{z}+\epsilon(t)I_{x}$, we have
\begin{equation}\label{state42}
\left(%
\begin{array}{c}
  \dot{x}_{t}^{A} \\
  \dot{y}_{t}^{A} \\
  \dot{z}_{t}^{A} \\
\end{array}%
\right)
=\left(%
\begin{array}{ccc}
  0 & -1 & 0 \\
  1 & 0  & -\epsilon(t) \\
  0 & \epsilon(t) & 0 \\
\end{array}%
\right) \left(%
\begin{array}{c}
  x_{t}^{A} \\
  y_{t}^{A} \\
  z_{t}^{A} \\
\end{array}%
\right),
\end{equation}
where $(x_{0}^{A}, y_{0}^{A}, z_{0}^{A})=(x_{0}, y_{0},
z_{0})=(\sin\theta_{0}\cos\varphi_{0},
\sin\theta_{0}\sin\varphi_{0}, \cos\theta_{0})$ and $\varphi_{0} \in
[0, 2\pi]$. From (\ref{state42}), we have
\begin{equation}\label{equation43}
\dot{z}_t|_{t=0}=\lim_{\Delta t\rightarrow 0}\frac{z^{A}_{\Delta
t}-z_{0}}{\Delta t}=\epsilon(0)y_{0}= \epsilon(0)
\sin\theta_{0}\sin\varphi_{0}.
\end{equation}
From (\ref{equation43}) and (\ref{bound}), it is easy to obtain the
following relationship:
\begin{equation}
\begin{array}{ll}
z_{\Delta t}^{A}-z_{\Delta t}^{B}& =z_{0}+\epsilon(0)
\sin\theta_{0}\sin\varphi_{0} \Delta t
-z_{0}[1-\frac{\epsilon^{2}(\Delta
t)^{2}}{2}]+\sin\theta_{0}\epsilon \Delta t +O((\Delta t)^{2})\\
& = \Delta t \sin\theta_{0}(\epsilon+\epsilon(0)\sin\varphi_{0})+O((\Delta t)^{2}).\\
\end{array}
\end{equation}

When $\theta_{0} \in [0, \pi)$, $\sin\theta_{0} \geq0$. Moreover, it
is always true that $\epsilon+\epsilon(0)\sin\varphi_{0}\geq 0$. If
$\theta_{0}\neq 0$ and $\epsilon+\epsilon(0)\sin\varphi_{0}\neq 0$,
we have
\begin{equation}
z^{A}_{\Delta t}> z^{B}_{\Delta t} .
\end{equation}

If $\theta_{0}=0$, we have
$(x'_{0},y'_{0},z_{0})=(x_{0},y_{0},z_{0})=(0,0,1)$. According to
Corollary \ref{corollary} and the proof of Theorem 2, we have
\begin{equation}
z^{A}_{\Delta t} \geq z^{B}_{\Delta t} .
\end{equation}

For $\epsilon+\epsilon(0)\sin\varphi_{0}= 0$, it corresponds to two
cases: (a) $\epsilon(0)=\epsilon$ and $\varphi_{0}=\frac{3\pi}{2}$;
(b) $\epsilon(0)=-\epsilon$ and $\varphi_{0}=\frac{\pi}{2}$. In the
following, we consider the case (a) (for the case (b) we have the
same conclusion as the case (a)). Since
$\varphi_{0}=\frac{3\pi}{2}$, $(x_{0}, y_{0}, z_{0})=(0,
-\sin\theta_{0}, \cos\theta_{0})$. Using a similar argument to the
proof of Theorem 2, we know that for $H^{A}=I_{z}+\epsilon(t)I_{z}$
with $(x_{0}, y_{0}, z_{0})=(0, -\sin\theta_{0}, \cos\theta_{0})$,
the optimal control for the performance index $J(\epsilon)=z_{f}$
takes a form of bang-bang control
$\epsilon(t)=\bar{\epsilon}=\epsilon$ or $-\epsilon$. So we only
need to consider a bang-bang strategy.

For such a bang-bang strategy as $H^{A}=I_z+\bar{\epsilon}I_x$, we
have
\begin{equation}
\left(%
\begin{array}{c}
  \dot{x}_{t}^{A} \\
  \dot{y}_{t}^{A} \\
  \dot{z}_{t}^{A} \\
\end{array}%
\right)
=\left(%
\begin{array}{ccc}
  0 & -1 & 0 \\
  1 & 0  & -\bar{\epsilon} \\
  0 & \bar{\epsilon} & 0 \\
\end{array}%
\right) \left(%
\begin{array}{c}
  x_{t}^{A} \\
  y_{t}^{A} \\
  z_{t}^{A} \\
\end{array}%
\right),
\end{equation}
where $(x_{0}^{A}, y_{0}^{A}, z_{0}^{A})=(x_{0}, y_{0}, z_{0})=(0,
-\sin\theta_{0}, \cos\theta_{0})$ and $\varphi_{0} \in [0, 2\pi]$.
We can obtain the corresponding solution as
\begin{equation}\label{oneA}
\left(%
\begin{array}{c}
  x_{t}^{A} \\
  y_{t}^{A} \\
  z_{t}^{A} \\
\end{array}%
\right) =\\
\left(%
\begin{array}{c}
  \frac{-\bar{\epsilon}\cos\theta_{0}}{1+\epsilon^{2}}\cos\omega t+\frac{\sin\theta_{0}}
  {\sqrt{1+\epsilon^{2}}}\sin\omega t+\frac{\bar{\epsilon}\cos\theta_{0}}{1+\epsilon^{2}} \\
  \frac{-\bar{\epsilon}\cos\theta_{0}}{\sqrt{1+\epsilon^{2}}}\sin\omega t-\sin\theta_{0}\cos\omega t \\
  \frac{\epsilon^{2}\cos\theta_{0}}{1+\epsilon^{2}}\cos\omega t-
  \frac{\bar{\epsilon} \sin\theta_{0}}
  {\sqrt{1+\epsilon^{2}}}\sin\omega t+\frac{\cos\theta_{0}}{1+\epsilon^{2}} \\
\end{array}%
\right).\\
\end{equation}

Now, we consider the limit as $\Delta t \rightarrow 0$ and obtain
\begin{equation}\label{comparision1}
\begin{array}{ll}
z^{A}_{\Delta t}-z^{B}_{\Delta t}&
=\frac{\epsilon^{2}\cos\theta_{0}}{1+\epsilon^{2}}
[1-\frac{(1+\epsilon^{2})(\Delta
t)^{2}}{2}]-\frac{\bar{\epsilon}\sin\theta_{0}}{\sqrt{1+\epsilon^{2}}}\sqrt{1+\epsilon^{2}}\Delta
t\\
&\ \ \
+\frac{\cos\theta_{0}}{1+\epsilon^{2}}-\cos\theta_{0}[1-\frac{\epsilon^{2}(\Delta
t)^{2}}{2}]+\sin\theta_{0}\epsilon \Delta t-\frac{\epsilon^{3}}{6}(\Delta t)^{3} \sin\theta_{0}+O((\Delta t)^{4})\\
& =\Delta t \sin\theta_{0}(\epsilon-\bar{\epsilon})+\sin\theta_{0}
(\Delta
t)^{3}(\bar{\epsilon}+\bar{\epsilon}\epsilon^{2}-\epsilon^{3})+O((\Delta
t)^{4})\\
& >0 .\\
\end{array}
\end{equation}
Hence, for arbitrary $z_{0}=\cos\theta_{0}$ ($\theta_{0}\in
[0,\pi)$), we have
\begin{equation}
z^{A}_{\Delta t}\geq z^{B}_{\Delta t} .
\end{equation}

For $t=\frac{\pi}{\epsilon}$, $z_{t}^{B}=-1$. Hence the relationship
$z_{t}^{A}\geq z_{t}^{B}$ is  always true.

We now define $g(t)=z^{A}_{t}- z^{B}_{t}$ and assume that there
exist $t=t_{1}\in [0, \frac{\pi}{\epsilon})$ such that
$z^{A}_{t_{1}}< z^{B}_{t_{1}}$. That is, $g(t_{1})<0$. Since $g(t)$
is continuous in $t$ and $g(0)=0$, there exists a time $t^{*}=\sup
\{t|0 \leq t<t_{1}, g(t)=0\}$ satisfying $g(t)<0$ for $t\in (t^{*},
t_{1}]$. However, we have proven that for any $z_{t}^{A}=z_{t}^{B}$
and $\Delta t \rightarrow 0$, $z_{t+\Delta t}^{A}\geq z_{t+\Delta
t}^{B}$, which contradicts $g(t)<0$ for $t\in (t^{*}, t_{1}]$.
Hence, we have the following relationship for $t\in [0,
\frac{\pi}{\epsilon}]$
\begin{equation}
z^{A}_{t}\geq z^{B}_{t}.
\end{equation}
\end{proof}


Now we can prove Theorem 1 using Lemma \ref{lemmaB} and Lemma
\ref{lemma2}.

\begin{proof}
For $H^{A}=I_{z}+\epsilon_{x}(t) I_{x}+\epsilon_{y}(t) I_{y}$, using
$\dot{\rho}=-i[H, \rho]$ and (\ref{eq4}), we
can obtain the following state equations
\begin{equation}
\left(%
\begin{array}{c}
  \dot{x}_{t}^{A} \\
  \dot{y}_{t}^{A} \\
  \dot{z}_{t}^{A} \\
\end{array}%
\right)
=\left(%
\begin{array}{ccc}
  0 & -1 & \epsilon_{y}(t) \\
  1 & 0  & -\epsilon_{x}(t) \\
  -\epsilon_{y}(t) & \epsilon_{x}(t) & 0 \\
\end{array}%
\right) \left(%
\begin{array}{c}
  x_{t}^{A} \\
  y_{t}^{A} \\
  z_{t}^{A} \\
\end{array}%
\right),
\end{equation}
where $(x_{0}^{A}, y_{0}^{A}, z_{0}^{A})=(0, 0, 1)$.

Define $\epsilon(t)=\sqrt{\epsilon_{x}^{2}(t)+\epsilon_{y}^{2}(t)}$
and $\epsilon_{x}(t)=\epsilon(t)\cos\gamma_{t}$ ,
$\epsilon_{y}(t)=\epsilon(t)\sin\gamma_{t}$. This leads to
the following equation
\begin{equation}\label{twouncertainties}
\left(%
\begin{array}{c}
  \dot{x}_{t}^{A} \\
  \dot{y}_{t}^{A} \\
  \dot{z}_{t}^{A} \\
\end{array}%
\right)=
\left(%
\begin{array}{ccc}
  0 & -1 & \epsilon(t) \sin\gamma_{t} \\
  1 & 0  & -\epsilon(t) \cos\gamma_{t} \\
  -\epsilon(t) \sin\gamma_{t} & \epsilon(t) \cos\gamma_{t} & 0 \\
\end{array}%
\right) \left(%
\begin{array}{c}
  x_{t}^{A} \\
  y_{t}^{A} \\
  z_{t}^{A} \\
\end{array}%
\right)
\end{equation}
where $(x_{0}^{A}, y_{0}^{A},
z_{0}^{A})=(\sin\theta_{0}\cos\varphi_{0},
\sin\theta_{0}\sin\varphi_{0}, \cos\theta_{0})$ and $\varphi_{0} \in
[0, 2\pi]$. From (\ref{twouncertainties}), we have
\begin{equation}\label{differential}
\dot{z}_t|_{t=0}=\lim_{\Delta t\rightarrow 0}\frac{z^{A}_{\Delta
t}-z_{0}}{\Delta t}=\epsilon(0) \cos\gamma_{0}y_{0}-\epsilon(0)
\sin\gamma_{0} x_{0}
\end{equation}
From (\ref{differential}) and (\ref{bound}), it is easy to obtain
the following relationship:
\begin{equation}
\begin{array}{ll}
z_{\Delta t}^{A}-z_{\Delta t}^{B}& =z_{0}+\epsilon(0)
\cos\gamma_{0}y_{0} \Delta t-\epsilon(0)\sin\gamma_{0} x_{0} \Delta
t\\
& \ \ \ -z_{0}[1-\frac{\epsilon^{2}(\Delta
t)^{2}}{2}]+\sin\theta_{0}\epsilon \Delta t +O((\Delta t)^{2})\\
& = \Delta t \sin\theta_{0}(\epsilon+\epsilon(0)\sin(\varphi_{0}-\gamma_{0}))+O((\Delta t)^{2}).\\
\end{array}
\end{equation}

When $\theta_{0} \in [0, \pi)$, $\sin\theta_{0} \geq0$. Moreover, it
is always true that
$\epsilon+\epsilon(0)\sin(\varphi_{0}-\gamma_{0})\geq 0$. If
$\theta_{0}\neq 0$ and
$\epsilon+\epsilon(0)\sin(\varphi_{0}-\gamma_{0})\neq 0$, we have
\begin{equation}
z^{A}_{\Delta t}> z^{B}_{\Delta t} .
\end{equation}

If $\theta_{0}=0$, we have
$(x'_{0},y'_{0},z_{0})=(x_{0},y_{0},z_{0})=(0,0,1)$. According to
Lemma \ref{lemmaB} and Pontryagin's minimum principle, we have
\begin{equation}
z^{A}_{\Delta t} \geq z^{B}_{\Delta t} .
\end{equation}

For $\epsilon+\epsilon(0)\sin(\varphi_{0}-\gamma_{0})= 0$, it must
be true that $\epsilon(0)=\epsilon$ or $-\epsilon$. We consider
$\epsilon(0)=\epsilon$ (for $\epsilon(0)=-\epsilon$ we have the same
conclusion as $\epsilon(0)=\epsilon$). Moreover, we have
$\gamma_{0}=\varphi_{0}+\frac{\pi}{2}$ or
$\gamma_{0}=\varphi_{0}-\frac{3\pi}{2}$.


For $\gamma_{0}=\varphi_{0}+\frac{\pi}{2}$, we first employ the fact
that $z_{t}$ is independent on $\varphi_{0}$ for $H=\epsilon
\sin\varphi_{0} I_{x}-\epsilon \cos\varphi_{0} I_{y}$ and $(x_{0},
y_{0}, z_{0})=(\sin\theta_{0}\cos\varphi_{0},
\sin\theta_{0}\sin\varphi_{0}, \cos\theta_{0})$ since
$z_{t}=\cos(\theta_{0}-\epsilon t)$. Then, it is easy to prove the
relationship $z^{A}_{\Delta t}\geq z^{B}_{\Delta t}$ using a similar
argument to that in the proof of Lemma \ref{lemmaB}. For
$\gamma_{0}=\varphi_{0}-\frac{3\pi}{2}$, we have the same conclusion
as in the case $\gamma_{0}=\varphi_{0}+\frac{\pi}{2}$. Thus, we
obtain
\begin{equation}
z^{A}_{\Delta t}\geq z^{B}_{\Delta t}.
\end{equation}

Now using a similar argument to that in the proof of Lemma
\ref{lemma2}, for arbitrary $t\in [0, \frac{\pi}{\epsilon})$, we
have
\begin{equation}
z^{A}_{t}\geq z^{B}_{t} .
\end{equation}

From (\ref{fidelity}), it is clear that the probabilities of failure
satisfy $p^{A}_{t}=\frac{1-z^{A}_{t}}{2}\leq
p^{B}_{t}=\frac{1-z^{B}_{t}}{2}$. That is, the probability of
failure $p_{t}^{A}$ is not greater than $p^{B}_{t}$ for $t\in [0,
\frac{\pi}{\epsilon})$. Hence, we can design the measurement period
$T$ using the case of $H^{B}$.

Using (\ref{fidelity}) and (\ref{bound}), for $t\in [0,
\frac{\pi}{\epsilon})$, we obtain the probability of failure
\begin{equation}\label{period2}
p_{t}^{B}=\frac{1-\cos{\epsilon}t}{2} .
\end{equation}
Hence, we can design the maximum measurement period as follows
\begin{equation}
T^{(1)}=\frac{\arccos(1-2p_{0})}{\epsilon}.
\end{equation}
\end{proof}

\section{CONCLUSIONS}\label{Sec5}

This paper proposes a variable structure control scheme with sliding
modes for the robust control of two-level quantum systems where an
eigenstate is identified as a sliding mode. We present a design
method for the control laws based on a Lyapunov methodology and
periodic projective measurements to drive and maintain this system's
state in the sliding mode domain. The key task of the control
problem is converted into a problem of designing the Lyapunov
functions and the measurement period. The Lyapunov function can be
constructed to define a control law. By using simulation, we obtain
an open-loop control to drive the controlled quantum system's state
into the sliding mode domain. For different situations of the
uncertainties in the system Hamiltonian, we give two approaches to
design the measurement period, which guarantees control performance
in the presence of the uncertainties. This sliding mode control
scheme provides a robust quantum engineering strategy for controller
design for quantum systems and has potential applications in state
preparation, decoherence control, quantum error correction
\cite{Dong and Petersen 2009NJP}, etc. Future work which can be
carried out in this area is listed as follows. 1) The physical
implementation of the proposed method on specific quantum systems.
For example, spin systems in NMR (see, e.g., \cite{Khaneja et al
2001}, \cite{Lee and Khitrin 2006}) may be a suitable candidate to
test the proposed approach. 2) The extension from two-level systems
to multi-level quantum systems: The basic idea of sliding mode
control can be extended in a straightforward way to multi-level
quantum systems. However, it is much more difficult to obtain an
analytical solution for a multi-level system. In \cite{Dong and
Petersen 2009NJP}, a numerical result has been obtained for a
three-level quantum system to determine the measurement period and
more complex systems are worth exploring by numerical methods. 3)
The extension to dissipative quantum systems governed by the
Lindblad equation \cite{Bonnard and Sugny} or described by a
stochastic differential equation: For such cases, it is necessary to
develop new methods to drive the system into the sliding mode domain
since the Lyapunov-based control approach does not usually work
\cite{Wang et al 2010}. 4) The exploration of practical applications
for the proposed approaches: The sliding mode may correspond to an
eigenstate or a state subspace and the sliding mode design approach
could be used in quantum state preparation and protection of encoded
quantum information in a subspace.

\section*{Appendix: Proof of $T^{(2)}\geq T^{(1)}$}
\begin{proof} Take $p_{0}$ as the variable and define
\begin{equation}\label{F_p0}
F(p_{0})=T^{(2)}-T^{(1)}
\end{equation}
For $p_{0} \in (0, \frac{\epsilon^{2}}{1+\epsilon^{2}})$, we have
the following relationship
\begin{equation}\label{f_p0}
f(p_{0})=F'(p_{0})=\frac{1}{\sqrt{\epsilon^{2}p_{0}-(1+\epsilon^{2})p^{2}_{0}}}
-\frac{1}{\sqrt{\epsilon^{2}p_{0}-\epsilon^{2}p^{2}_{0}}}
\end{equation}

It is clear from (\ref{F_p0}) and (\ref{f_p0}) that $f(p_{0})>0$ for
$p_{0} \in (0, \frac{\epsilon^{2}}{1+\epsilon^{2}})$ and
$F(p_{0}\rightarrow 0^{+})=0$. Hence $F(p_{0})\geq 0$ for $p_{0} \in
(0, \frac{\epsilon^{2}}{1+\epsilon^{2}})$.

When $p_{0}=p'=\frac{\epsilon^{2}}{1+\epsilon^{2}}$,
\begin{equation}
T^{(1)}{(p')}=\frac{\arccos{(\frac{1-\epsilon^{2}}{1+\epsilon^{2}}})}{\epsilon}
,
\end{equation}
\begin{equation}
T^{(2)}{(p')}=\frac{\pi}{\sqrt{1+\epsilon^{2}}} .
\end{equation}

Let $x=\frac{1-\epsilon^{2}}{1+\epsilon^{2}}$ and
\begin{equation}
G(\epsilon)=\frac{\epsilon
\pi}{\sqrt{1+\epsilon^{2}}}-\arccos{(\frac{1-\epsilon^{2}}{1+\epsilon^{2}}}).
\end{equation}
We have
\begin{equation}
\widetilde{G}(x)=\frac{\pi}{\sqrt{2}}\sqrt{1-x}-\arccos{x} .
\end{equation}
For $x\in [-1,\ 1]$, $\widetilde{G}(x)$ is continuous in $x$ and we
also know that $\widetilde{G}(x)=0$ only when $x=\pm 1$. It is easy
to check $\widetilde{G}(x=0)>0$. Hence, we know that for $x\in [-1,\
1]$, $\widetilde{G}(x)\geq 0$. That is, for $\epsilon >0$,
$G(\epsilon)\geq 0$. From the relationship $G(\epsilon)\geq 0$, we
know $T^{(2)}{(p')}\geq T^{(1)}{(p')}$ for $\epsilon >0$.

Hence we concluded that for arbitrary $p_{0} \in (0,
\frac{\epsilon^{2}}{1+\epsilon^{2}}]$, $T^{(2)}\geq T^{(1)}$.
\end{proof}

\section{ACKNOWLEDGMENTS}

The authors would like to thank anonymous reviewers and the
Associate Editor for helpful comments and suggestions.



\end{document}